\documentclass[11pt, oneside]{article}   	
\usepackage{textcomp}
\usepackage{geometry}                		
\usepackage{graphicx}				
\usepackage{caption}
\usepackage{subcaption}								

\usepackage{amssymb}
\usepackage{amsthm}
\newtheorem{theorem}{Theorem}
\newtheorem{definition}{Definition}
\usepackage{natbib}
\usepackage{hyperref}
\usepackage{authblk}
\hypersetup{
    colorlinks=false,
    linkcolor=blue,
    filecolor=magenta,      
    urlcolor=blue,
    pdfpagemode=FullScreen,
}
\usepackage{float}

\title{Center-Periphery Structure in Communities: Extracellular Vesicles}
\author[1]{Eleanor Wedell\thanks{EW and MP contributed equally to this manuscript}}
\author[1]{Minhyuk Park}
\author[2]{Dmitriy Korobskiy}
\author[1]{Tandy Warnow\thanks{warnow@illinois.edu}}
\author[1,3]{George Chacko\thanks{chackoge@illinois.edu}}
\affil[1]{Department of Computer Science, University of Illinois Urbana-Champaign, Urbana, IL 61801}
\affil[2]{NTT DATA, McLean, VA, 22102}
\affil[3]{Office of Research, Grainger College of Engineering, University of Illinois Urbana-Champaign, Urbana, IL 61801}

\begin{document}
\maketitle

\begin{abstract}

\noindent  
Clustering and community detection in networks are of broad interest and have been the subject of extensive research that spans several fields. We are interested in the relatively narrow question of detecting communities of scientific publications that are linked by citations. These publication communities can be used to identify scientists with shared interests who form communities of researchers. Building on the well-known k-core algorithm, we have developed a modular pipeline to find publication communities. We compare our approach to communities discovered by the  widely used Leiden algorithm for community finding.  Using a quantitative and qualitative approach, we evaluate community finding results on a citation network consisting of over 14 million publications relevant to the field of extracellular vesicles.
\end{abstract}

\section{Introduction} 

Community detection in networks is of broad interest and has been discussed in comprehensive reviews \citep{Fortunato2009,Fortunato2010,Javed2018}. At a high level, the community detection problem amounts to identifying groups within a complex network that share some common properties.  As observed by \cite{Coscia2011}, this definition suffers from some degree of imprecision given the diverse ways in which communities can be defined and a richness of perspectives. For example, a community detection approach may focus on vertex similarity or edge-density; disjoint, overlapping or hierarchical community structure; and static versus dynamic networks.  \cite{Newman2006} had also noted distinctions  between community discovery and graph partitioning. Thus, the context of the question being asked and the techniques being employed tend to determine the flavor of community detection in a study.

From the perspective of scientometrics, detecting a research community can be framed as a community finding problem in which communities of publications defining areas of research are discovered in the scientific literature. First, the scientific literature is modeled as a network with publications as nodes and citations as directed edges \citep{Boyack2019}. 
From this network,  an area of research is defined by a community of publications--a sufficiently citation-dense area in the network. Citation networks of scientific literature can be constructed using different approaches. Direct citations were used to build clusters of articles from a dataset of over 10 million publications \citep{Waltman2012}; this methodology was also used in building citation maps from 19 and 43 million publications \citep{Boyack2014}. 

Direct citation, bibliographic coupling, and co-citation have been compared for their relative value in identifying research fronts \citep{Boyack2010}, with a hybrid approach involving bibliographic coupling and textual similarity performing the best.  A subsequent study conducted at a larger scale and with improved evaluation criteria suggested that direct citation was the most promising \citep{Klavans2017}. We use direct citations in this study.

Once constructed, citation networks can be analyzed using different community finding or clustering approaches \citep{Boyack2019,Ahlgren2020,Traag2019,LeoneSciabolazza2017,Waltman2012,Subelj2016}. Of these approaches, a recent development is the availability of the Leiden algorithm, which offers better partitioning and performance \citep{Traag2019} compared to a preceding approach, the Louvain algorithm \citep{Blondel_2008}, which seeks to optimize the modularity quality function \citep{Newman2006}. 

We are interested in how research communities form and collaborate around research questions \citep{Kuhn1970,crane1972invisible}. Collaboration within the scientific community goes beyond co-authorship, and includes building upon prior discovery by other researchers; citations between publications thus indicate this more general form of collaboration, and the publication communities detected in this process are evidence of such collaborations. Further, the authors of the publications in the publication communities represent communities of researchers working on related questions. Historical studies have estimated the size of research communities to be in the order of a few hundreds \citep{Price1966,crane1972invisible,Kuhn1970,Mullins_1985}.  However, this question merits re-examination in the modern scientific enterprise; expanded, globalized, and electronically connected.

Since researchers can work on more than one problem and be considered members of more than one research community,  to find author communities we begin by finding communities of publications. Then, for each publication community, the authors of the publications in the community constitute  a researcher community organized around the questions defined by the publication community. 

We are interested in articles linked by citation as the products of a research community rather than articles clustered by textual similarity, therefore, we use direct citations to discover communities. In an earlier study \citep{Chandrasekharan2021}, we explored this approach  to detect publication communities and subsequently their author communities in networks of biology literature. We used an ensemble technique to find publication communities  coupled with limited qualitative analysis, where the publication communities were  significantly overlapping  in clusters identified by the Leiden algorithm and also the Markov Clustering (MCL) algorithm \citep{VanDongen2008}.

We were specifically interested in whether the author communities we found exhibited substructure indicating the influence of  a few individuals associated with the majority of publications while exhibiting different degrees of collaboration within and across subgroups \citep{Price1966,crane1972invisible}. While we did detect such center-periphery substructure \citep[p.~60]{Breiger2014} in both publication and author communities, our findings were potentially limited by the clustering methods we used, Leiden and MCL, since neither is designed to detect or require substructure when identifying communities.  

Here, we aim to investigate more carefully whether publication clusters exhibiting center-periphery structure exist in citation networks of the scientific literature. The central idea is that each community contains core nodes representing the ``center" of center-periphery organization and additional nodes representing the ``periphery''; our approach first finds the core nodes and then augments the cluster to include peripheral nodes.   Since prior community detection methods are not explicitly designed to detect such communities, we propose a new modular pipeline that specifically aims to find such communities. To find the core node sub-community, we combine two approaches represented in the clustering and community detection literature: first, that valid communities should have positive modularity score \citep{Fortunato2006}, and second that each community should be {\em dense}, which is expressed by the ratio between the average node degree  and the number of nodes in the network. 

Thus, our approach combines three ideas from the literature: the center-periphery model of authorship communities extended here to a model for publication communities, positive modularity for individual communities, and sufficient citation density for individual communities. 
By distinguishing between  the core and non-core nodes, we can require that the positive modularity and sufficient citation density requirements hold for the subclusters of core nodes but not necessarily for the clusters that combine both core and non-core nodes. This  distinction potentially enables us to better model real world community structure.

Our modular pipeline begins by finding clusters of core nodes, building on the ideas of \cite{Giatsidis2011} who quantified the cohesiveness of a community using the ``k-core" concept from graph theory \citep{matula1983smallest}. Although optimizing the total modularity in the clustering has drawbacks (as demonstrated in \cite{Fortunato2006}), we also required that the clusters  individually have modularity scores above $0$; this is a relatively mild criterion that seeks internal cohesion, and has been considered in the prior literature  (e.g., \cite{Newman2004,Fortunato2006} to be evidence of a valid community. 

We tested our pipeline on a network of over 14 million publications that we constructed by harvesting citing and cited articles from a seed set defined by the keyword  ``exosome". This keyword captures articles from the field of extracellular vesicles, which may be important for intercellular communication and development of some diseases, as well as having potential for therapy \citep{edgar2016,kalluri2020,raposo2021}. We chose extracellular vesicles~\citep{raposo2021} as the focus of this study for two reasons: first, it is a large research area, and second, because it is rapidly expanding--it  has seen spectacular numbers of publications each year since 2010, and therefore represents an excellent test case for community finding methods in the modern scientific enterrprise.

Expecting that not all communities discovered in such a large network would be directly relevant to exosomes, we use 1,218 cited references from 12 relevant review articles as expert-identified markers for specificity in the communities we discover. We report our findings in the following sections. 

\section{Materials and Methods}
\subsection{Data}
A citation network consisting of 14,695,475 nodes and 99,663,372 edges was generated using the Dimensions bibliography \citep{hook2018dimensions}. Briefly, a ``seed'' set, $S$, was obtained by performing a text search for the term ``exosome'' with years of publication restricted to 2010 or earlier. This constraint was applied to allow every element in the seed set to have accumulated at least 10 years of citations. The search retrieved 11,156 publications of type article from Dimensions. To capture publications proximal by citation to the seed set, a network was constructed using a protocol we labelled SABPQ.  First we start with the seed set $S$. Set $A$ is the set of publications that cite at least one publication from set $S$. Similarly, set $B$ is the set of publications that are cited by at least one publication from set $S$. Once the sets $S$, $A$, and $B$ are identified, then we define set $P$ as the set of publications that cite at least one publication from the set $S \cup{} A \cup{} B$.  Similarly, set $Q$ is the set of publications that are cited by at least one publication from the set $S \cup{} A \cup{} B$. The  network contains directed edges defined by citations; if publication $x$ cited $y$ then we created an edge from $x$ to $y$. The SABPQ protocol was implemented using Dimensions on BigQuery in Google Cloud Services. The data was then exported to Google Cloud Storage Bucket and subsequently exported to a PostgreSQL database for further analysis.

\emph{Marker Nodes and Specificity} As marker nodes for our analysis, we used 1,218 articles cited in 12 recent reviews on extracellular vesicles  and exosome biology~\citep{vNiel2018,kalluri2020,verdi2021,ghoroghi2021,lananna2021,busatto2021,he2021,schnatz2021,lelay2021,leidal2021,clancy2021,raposo2021}. All 1,218 markers are present in our constructed network with 14,695,475 nodes.  Any community containing at least one marker node was considered relevant. Marker nodes were matched to clusters using identifiers in our network. For example, the marker node with title ``Tumour exosome integrins determine organotropic metastasis"  and DOI 10.1038/nature15756  is identified as node 4431204 in our network. The complete list of marker nodes is available on our Github site \citep{Park2021}. 

\subsection{Clustering methods}

\paragraph{Leiden}
We used version 1.1.0 of the Java implementation for the Leiden algorithm \citep{Traag2019} provided by the Centre for Science and Technology Studies and available in Github \citep{leiden-github}. 
We ran  Leiden in default mode, which means that the quality function being optimized was the Constant Potts Model rather than modularity. 
Leiden includes a parameter for the resolution value, which we vary in our experiments from 0.05 to 0.95.

\paragraph{New clustering methods}

Here we describe at a high level the clustering and community-finding methods we developed in our study. The methods we developed and use in this study are freely available in Github, and the locations of these software and exact commands we used are provided in the supplementary materials. We also provide software version numbers and commands for the existing codes that we use in the supplementary materials. We note here that the codes we developed rely on NetworKit  \citep{networkit2016}, which is an open-source Python module designed for scalable network analysis.

In our approach, the objective was to produce a set of clusters each of which has core nodes and peripheral nodes, consistent with the ``center-periphery" structure described earlier.  These clusters are considered to be ``publication communities", with two types of members: core members that are densely connected to each other and peripheral (non-core) members connected to the core members but with fewer edges within a cluster. 

Our approach requires values for $k$ and $p$, where $k$ specifies a minimum connectivity between the core nodes, and $p$ indicates a minimum connectivity between each non-core (``periphery") node and the core nodes.
These parameter values for $k$ and $p$ are provided by the user, and different choices for these parameters will produce different clusterings. 

The required minimum connectivity $k$ between core nodes is related to the $k$-core concept in graph theory, which we now describe. A $k$-core of a network $N$ is a largest connected subnetwork $A$ of $N$ such that every node in $A$ is adjacent to at least $k$ other nodes in $A$  \citep{seidman1983,matula1983smallest,pittel1996}. The $k$-cores can be calculated in polynomial time \citep{matula1983smallest}, and our new clustering methods build on these algorithms.

We are also interested in the modularity scores of the clusters that are produced by each method, as
given in Definition \ref{def:modularity}:

\begin{definition}
\label{def:modularity}
The modularity of a single cluster $s$ within a network $N$, denoted by mod(s), is
given by
$$mod(s) = \frac{l_{s}}{L} -  \left ( \frac{d_{s}}{2L} \right )^{2}$$ where
 $l_{s}$ is the number of internal edges in cluster $s$, $d_{s}$ is the sum of the degrees of the nodes inside $s$, and $L$ is the number of total edges in the network $N$ \citep{Fortunato2006}.
 The total modularity of a clustering is the sum of the modularity scores of its clusters.
 \end{definition}

Rather than aiming to maximize the total modularity of the
clustering, we will only require that each cluster have positive modularity; as noted
in \cite{Fortunato2006}, this approach aims to detect valid communities.
We now define some additional terms that we will use in designing new  clustering methods:
\begin{definition}
\label{def:valid}
Given a network $N$, a clustering $\mathcal{C}$,  and a cluster $C$ drawn from $ \mathcal{C}$ where $C$ is partitioned into
core nodes   and non-core nodes, 
we will say:
\begin{itemize}
\item $C$ is $k$-valid if and only if each core node in $C$ is adjacent to at least $k$ other core nodes in $C$
\item  $C$ is $m$-valid if and only if the sub-cluster induced by the core nodes  is connected and 
has a positive modularity score
\item $C$ is $p$-valid if and only if each non-core node is adjacent to at least $p$ core nodes in $C$
\item $C$ is $kmp$-valid if and only if it is $k$-valid, $m$-valid, and $p$-valid
\item The clustering $\mathcal{C}$ is $kmp$-valid if and only if every cluster $C \in \mathcal{C}$ is $kmp$-valid
\end{itemize}
Note that if a cluster does not contain non-core nodes then it is {\em vacuously} $p$-valid.
\end{definition}

The clustering methods that we develop seek to produce $kmp$-valid clusters, so that we can interpret these clusters as 
communities with center-periphery structure. Moreover, we are interested in clusterings that produce a large number of $kmp$-valid clusters, as well as those that include as many nodes as possible in the $kmp$-valid 
clusters (which by definition must be non-singleton when $k>1$). 
Hence we explore different techniques that seek to optimize  these two opposing criteria. 

We also require positive modularity in the  core node sub-clusters for each cluster. 
This is  a relatively mild requirement that avoids 
cases where the core node sub-cluster may be k-valid and connected but might not reflect a preference for itself over the outside.
Consider the case where a 10-clique, a complete graph on 10 nodes, is contained in a clique with 20 nodes.
This 10-clique would satisfy k-validity for $k \leq 9$ and would be connected, but would not have positive modularity.
By enforcing positive modularity, we would avoid returning such clusters. This example illustrates the advantage of enforcing 
positive modularity even though the probability of it occurring in a real world network is likely to be small. 
We also note that enforcing positive modularity in the core node sub-cluster (or even in the final cluster that contains both core and non-core nodes) is 
not the same as trying to maximize the sum of the modularity scores of the individual clusters  (total modularity score).
In other words, enforcing positive modularity does not have the same vulnerability  to the resolution limit that was established for the modularity criterion, which 
seeks to maximize  the total modularity score \citep{Fortunato2006}.

\subsubsection{Four-Stage $kmp$-Clustering}
We designed a four-stage pipeline that is designed to enable the user to explore different clustering options, while guaranteeing that the output is a $kmp$-valid clustering.
The input is a network $N$ and values for the parameters $k$ and $p$.
At a high level,  the first stages aim to construct the core member components. The second stage extracts valid sub-clusters from those generated by the first stage.  
The third stage augments these clusters with additional members, most likely non-core members, though some might qualify as core members, and the fourth stage assigns core or non-core status to the nodes, and retains only those clusters that are $kmp$-valid. 
We begin with a description of the overall multi-stage structure of our new clustering methods; note that stages 2 and 3 are optional.

\begin{itemize}
\item Stage 1: Cluster the network $N$ into disjoint clusters (core members), so that each non-singleton cluster  is $k$-valid and $m$-valid.
\item Stage 2: Attempt to break each non-singleton cluster produced in Stage 1 into a set of  pairwise disjoint clusters, each of which is $k$-valid at minimum.
\item Stage 3: For each non-singleton cluster, add unclustered nodes, nodes that are not in any non-singleton cluster, as non-core (peripheral) members, provided that they 
are adjacent to at least $p$ core nodes in the selected cluster. This is the {\em augmentation} stage, which adds non-core nodes to the clusters produced in the earlier stage.
\item Stage 4: Process the clustering that is received so that each final cluster is partitioned into core and non-core members, and so that the clustering is $kmp$-valid.
\end{itemize}

Thus, Stages 1 and 2 together are directed at finding core members of clusters, with Stage 1 directed at clusters with large numbers of core nodes and Stage 2 aimed at extracting smaller clusters within these larger clusters.  
At  the end of Stage 1, all clusters will be $k$-valid and $m$-valid.  If the optional Stage 2 is applied,  the clusters it produces will be $k$-valid and connected, but may no longer have positive modularity.
Neither of these stages introduces any non-core nodes, and so the output of each stage is vacuously $p$-valid. 
 Stage 3  augments the clusters to include non-core nodes;  by design the clusters will be connected, $k$-valid, and $p$-valid, but  depending on the outcome of
Stage 2, they may not have positive modularity and so may not be $m$-valid. 
Stage 4 is designed to ensure that all final clusters are $kmp$-valid, and so may modify or discard clusters found in the earlier stages.
However, after Stage 4 is run,  the output clustering is guaranteed to be $kmp$-valid. Furthermore,   the clusters produced are parsed into core and non-core nodes.

We now describe the techniques we have developed for each stage.
For Stage 1, we present iterative k-core (IKC) clustering, a method that is inspired by the k-core concept in graph theory.
For Stage 2, we also present two different techniques: recursive Graclus (RG) and iterative Graclus (IG), both of which are based on the Graclus \citep{Dhillon2007} clustering method used in its default setting and applied to split a graph into two parts.
Stage 3 is implemented using a straightforward algorithm that we describe below.
 In contrast to the earlier stages, Stage 4 involves multiple steps, and is described below.  
 This multi-stage design provides a flexible framework by allowing different techniques to be used at each stage.

\subsubsection{Stage 1: Iterative k-core clustering (IKC)} To motivate the IKC algorithm, we first describe the technique that computes $k$-cores and then describe the iterative method.

\paragraph{$k$-core} 

For a network $N$ and specified positive integer value for $k$, a $k$-core of a network $N$ is the largest connected subnetwork $A$ of $N$ such that every node in $A$ is adjacent to at least $k$ other nodes in $A$.
Note that  for every network that does not have any isolated vertices (i.e., nodes of degee $0$),  each connected component is a $1$-core of the network. The distribution of $k$-core sizes also indicate how quickly a network shrinks as $k$ increases \citep{Leskovec_2008}.

The identification of the $k$-core  for a maximum achieved value of $k$ in a network has been proposed as a quality measure for community finding that measures cohesiveness and suggests collaboration \citep{Giatsidis2011}. This idea is reiterated in \cite{kong2019,malliaros2019}, who discuss applications of the $k$-core  in biology and real world networks. 

 The $k$-cores can be calculated in polynomial time \citep{matula1983smallest}, as follows. 
First, we calculate the degree of every node in the network. Then, every node of degree less than $k$ is deleted from the graph, and this process repeats until every node has degree at least $k$ (i.e., every node
is adjacent to at least $k$ other nodes).  Every connected component that remains is called a $k$-core of the network.  
As an example, given a network that contains  two connected components: one is a clique of size $100$ and the other has a node $x$ that is adjacent to 2000 other nodes,  each of which is only adjacent to $x$ (and so have degree $1$). 
Note that for all $k$ with $2 \leq k \leq 99$,  the $k$-core of this network  is the clique of size 100.

\paragraph{$k$-core clustering}
  The simple $k$-core clustering method takes as input a network $N$ and
a value for $k$, and computes the $k$-cores of the network.
The set of the $k$-cores is returned as the clustering. 
By construction, the simple $k$-core clustering method produces clusters that are $k$-valid and connected.
However, it does not constrain the clusters to have positive modularity.

\paragraph{Iterative k-core clustering}
To improve on the simple k-core clustering technique for our purposes, we developed an iterative k-core (IKC) algorithm. 
The  input  to the IKC algorithm is a network $N$ and a positive integer $k$. IKC then operates as follows:
\begin{itemize}
\item We will construct a bin $B$ of clusters that will be returned by the IKC clustering. 
In this step, we initialize $B$ to be the empty set. 
\item 
We run the k-core labelling algorithm, which labels every node in the network with a non-negative integer. 
We let $L$ be the largest label found in this labelling. 
If $L < k$, then IKC exits, and returns the clusters in the bin $B$.Otherwise, the $L$-cores (i.e., the connected components that are labelled by $L$)  of the network are  evaluated as  potential clusters.
\item 
An $L$-core $A$ is added to the bin $B$ if and only if $A$ has positive modularity.
\item The $L$-core is then deleted from the network, 
and the residual network is recursively analyzed by IKC. The stopping condition is when all the nodes have been deleted from the network.
\end{itemize}

This procedure produces a collection of clusters, each of which has the following properties:
 (i) each cluster is connected and has positive modularity, and hence each cluster is $m$-valid and
(ii) each cluster is $k$-valid, which means every node in the cluster is adjacent to at least $k$ other nodes in the cluster. 

\subsubsection{Stage 2: Finding clusters within clusters using Graclus} In Stage 2  we seek to discover $k$-valid clusters that exist within larger $k$-valid clusters. 
For this, we use the Graclus clustering method \citep{Dhillon2007}. 
Graclus has been previously used to cluster bibliometric data \citep{Dhillon2007,devarakonda2020,Subelj2016}, but here we use it to split a given cluster into two sub-clusters. 
We use Graclus to optimize its default criterion, which is the normalized cut criterion. 
In this setting, we seek a partition of a given cluster $C$ into two parts $C_1$ and $C_2$ so as to 
minimize $${{links(C_1,C_2)}\over{links(C_1,C)}}+ {{links(C_1,C_2)}\over{links(C_2,C)}},$$
where $links(A,B)$ denotes the number of edges with one endpoint in $A$ and the other endpoint in $B$.

As is the case in other optimization methods, local search in Graclus helps the optimizer escape poor local minima; its default setting does no iterations but this can be modified by specifying the number of iterations. In this study we explore both the default mode, $l=0$ (no iterations) and $l=2000$ iterations.  
We implement our use of Graclus in two different ways: recursively and iteratively. 
Thus, we use Graclus in four different ways.

\paragraph{Recursive Graclus} The input to this method is a clustering of  the network $N$ and the value for the parameter $k$.
We create  a bin $B$ of clusters (setting it initially to the empty set).
We take a cluster $C$ from the clustering and apply Graclus recursively using  either the default mode or the local search mode.

The result of this is a division of the cluster $C$  into two non-empty sets $A_1$ and $A_2$.
If $A_1$ has positive modularity and is $k$-valid, then we add $A_1$ to $B$ (the bin we have created), and similarly for
subset $A_2$.
If neither $A_1$ nor $A_2$  is added to   $B$, add cluster $C$ to $B$ and delete $C$ from the network, effectively removing it from further consideration by Recursive Graclus.

The stopping condition for Recursive Graclus is that all nodes have been deleted from the network.
When the stopping condition is reached, the final output of Recursive Graclus is the set of clusters in bin $B$.
Though all clusters that are produced by Recursive Graclus are guaranteed to be $k$-valid  and have positive modularity, they may not
be  $m$-valid since this requires that the core node subsets be connected. 

\paragraph{Iterative Graclus} To use Graclus iteratively, we follow a similar procedure as in Recursive Graclus but with two key differences. The first is that the user provides a parameter for the number of iterations, so that the procedure must stop after that number of iterations, if it hasn't already stopped. The second difference is that, in contrast to Recursive Graclus,  the procedure is guaranteed to produce clusters that are $k$-valid and $m$-valid in each iteration, as we now describe.  

When we apply Graclus, in either default or local search mode, to split a cluster into two sub-clusters, each of the created sub-clusters is parsed into core and non-core nodes (using a variant of the algorithm described for Stage 4, see Section \ref{sec:kmp-other}). 
If the core node set is empty, the sub-cluster will be discarded.  However, if the core is non-empty then the core node set is by definition $k$-valid, and is then evaluated further. Each core node set is divided into its connected components, and each connected component that has positive modularity is passed to the next iteration.
Any cluster that does not end up producing a sub-cluster that is passed to the next iteration is added to the bin $B$ as in Recursive Graclus and is then deleted from the network.  Iterative Graclus stops when one of two conditions occurs: the number of allowed iterations has been completed, or all the nodes have been deleted from the network. 
By design, the output of Iterative Graclus is a set of clusters each of which is $k$-valid and $m$-valid (thus, every cluster is connected, has positive modularity, and every node is adjacent to at least $k$ nodes in its cluster).

\subsubsection{Stage 3: Augmentation} The purpose of the augmentation step is to assemble the periphery of center-periphery structures. 
The input to Stage 3 is a set of clusters, so that each non-singleton cluster is $k$-valid and $m$-valid. 
Here we allow all nodes that are not in any non-singleton cluster to be added to some cluster as long as it is adjacent to at least $p$ core nodes in the (single) cluster to which it is added. In this study, we set $p = 2$ to ensure that we captured publications that 
are linked by co-citation  or bibliographic coupling to core nodes in a community.  If no such cluster exists such that a node can be added to it in a $p$-valid manner, the node remains unclustered. 

We add $x$ to the  cluster $C$ that maximizes $$\frac{N_C(x)}{|C|}$$ where $N_C(x)$ is the number of core node neighbors of node x in cluster $C$ and where $|C|$ denotes the number of nodes in cluster $C$. 
In other words, we add node $x$ to the cluster where $x$ has proportionally the most core node neighbors.

As an example, suppose $C_1$ and $C_2$ are clusters of core nodes and that  $x$ is not yet added as a non-core node to any cluster.
Suppose $x$ has 5 neighbors in cluster $C_1$ and 10 nodes in cluster $C_2$, where $|C_1|=1000$ and $|C_2|=20$. This procedure would add $x$ as a non-core
member to $C_2$ since 50\% of the nodes in $C_2$ are neighbors of $x$ while only $5/1000 = 0.5\%$ of the nodes in $C_1$ are neighbors of $x$. 

\subsubsection{Stage 4: Parsing clusterings to produce $kmp$-valid clusters} 

Although Stage 2 is guaranteed to produce $k$-valid clusters, it does not always produce $m$-valid clusters. 
Furthermore, the impact of Stage 3 (the augmentation step) is to add nodes to clusters that can participate as
non-core nodes, and
it is  possible for a node added to a cluster in Stage 3 to have sufficient neighbors in its cluster
to qualify for core membership. 
Hence, the result of these three stages is  a set of clusters that needs to be
``parsed" in order to know which nodes are core members, which nodes are non-core members, and whether
the clusters are $kmp$-valid (as defined in Definition \ref{def:valid}).

\paragraph{Parsing and modifying a single cluster}

\noindent
Here we describe how we perform this parsing on a given
cluster $C$ (taken from a clustering  $\mathcal{C}$) and values for $k$ and $p$. 
 \begin{itemize}
 \item 
Step 1: We label every node in the cluster $C$ using the k-core labelling algorithm, applied only to the subnetwork defined by $C$.
We let $C'$ denote the subset of nodes in $C$, each of whose labels is at least $k$, and we put the remaining nodes
 into a bin $B(C)$.
 \item 
 Step 2:
We compute the connected components of $C'$, and delete the components that do not have
positive modularity; the retained components thus have positive modularity, and are referred to as $C$-derived clusters (to indicate their derivation from the original cluster $C$).
The clusters produced in this step will be the core node members in the final clusters we produce 
at the end of Step 3.
\item  
Step 3: We augment the $C$-derived clusters using the bin $B(C)$ (i.e., we find non-core members to add) as follows.
We examine each node $x$ in bin $B(C)$ to see if it has at least $p$ neighbors in at least one
$C$-derived cluster; 
if so,  we select the best $C$-derived cluster  $A$ (using the algorithm from Stage 3) 
and add the node $x$ to $A_{nc}$ (where ``nc" refers to ``non-core").
After all nodes are examined and processed, we let 
 $A'=A \cup A_{nc}$ (for each $C$-derived cluster $A$) and output
 the set of all such clusters $A'$ as the ``final
clusters" derived from cluster $C$.
Note that this process indicates the
parsing of each final cluster $A'$ into core ($A$) and non-core ($A_{nc}$) nodes.
\item 
Step 4: We return all the final  clusters derived from $C$.
\end{itemize}

\begin{theorem}
For any clustering $\mathcal{C}$ of a network $N$, and any positive integers for $k$ and $p$ (with $p < k$),
the output of $kmp$-processing is a clustering that is $kmp$-valid.
Therefore, the output of the four-stage clustering method is $kmp$-valid for all networks $N$ and values for $k$ and $p$.
\label{theorem}
\end{theorem}
\begin{proof}
Let  $\mathcal{C}$ be an arbitrary clustering. Hence, some of its clusters
may not be $k$-valid, may not be connected, and may not have positive modularity.
We will prove that after the $kmp$-processing, all the clusters are $kmp$-valid.
Specifically, we will prove that the parsing produced in Step 3 into core and non-core
satisfies $kmp$-validity.

First, note that by construction the clusters (referred to as $C'$) that are produced in Step 1 have the property that every node 
in these clusters is adjacent to at least $k$ other nodes in their cluster.
Hence, treating each cluster as only containing core nodes, these clusters
are $k$-valid.
In Step 2, these clusters are divided into components and the components are retained only if they have positive modularity; hence
the $C$-derived clusters  that are produced are $k$-valid and $m$-valid, under the interpretation that they contain only core nodes.
In Step 3, the $C$-derived clusters are augmented. This augmentation step maintains
connectivity, and so the final clusters are connected. 
It remains to establish that after parsing  any final cluster into core and non-core nodes, the final cluster would 
be $k$-valid (i.e., every core node would be adjacent to at least $k$ other core nodes), $p$-valid
(i.e., every non-core node would be adjacent to at least $p$ core nodes), and the core node subcluster would have positive modularity, and hence also be $m$-valid.

Let $A'$ be some final cluster.  By construction, it is formed by taking a $C$-derived cluster $A$ produced in Step 2, and then augmenting it.
Thus, $A' = A \cup A_{nc}$, where $A_{nc}$ is the set of nodes that are added during the augmentation step.
We will establish that applying Stage 4 $kmp$-parsing to this cluster $A'$ will not change its
decomposition into core and non-core (i.e., $A$ will still be the core nodes and $A_{nc}$ will be
the non-core nodes). 
Note that by construction, all the nodes in  $A$ are adjacent to at least $k$ other nodes in $A$.
Hence, when $A'$ is $kmp$-parsed, the nodes that are identified as core nodes will contain all the nodes in $A$ and then possibly
some nodes from  $A_{nc}$. 
Independent of whether there are new core nodes, $A'$ will be $p$-valid.
If there are no new core nodes, therefore, then $A'$ will be $kmp$-valid. 
Here we show that in fact no node in $A_{nc}$ will be labelled as core, and so
there are no new core nodes. 

Let $x \in A_{nc}$ with $A$ a $C$-derived cluster.
Hence, $x$ is drawn from bin $B(C)$.
By Step 1, 
the label assigned to $x$ during Step 1 (when the nodes in cluster $C$ were labelled) was a value $L_1 $ that is strictly less than $k$.  
Since $A$ is a $C$-derived cluster, $A'\subseteq C$. 
Consider the label $L_2$ assigned to $x$ by the k-core labelling of $A'$.  
Since $A' \subseteq C$, it follows that $L_2 \leq L_1$. Since
$L_1 < k$, it follows that $L_2 < k$. 
Hence, $x$ will not be placed in the core for $A'$, and the final clusters output in Step 4 are therefore $kmp$-valid. 
\end{proof}

\subsubsection{Additional uses of $kmp$-parsing}
\label{sec:kmp-other}

The Stage 4 $kmp$-parsing routine is also used in modified forms in other aspects of this study:
\begin{itemize}
\item Strict $kmp$-parsing: 
This strict parsing routine  evaluates each cluster in a clustering for being kmp-valid: those clusters that
are $kmp$-valid are retained and all others are discarded.  We use this strict parsing routine to evaluate Leiden. 
\item Using $kmp$-parsing to extract core node clusters: 
We also have a variant where use the parsing to extract only the core nodes within each cluster, and then return the components of the core node  sub-clusters that have positive modularity score; this variant is used within Iterative Graclus.
\end{itemize}

\subsection{Multidimensional Scaling Analysis} We used Multidimensional Scaling analysis (MDS) to visualize clusters of marker nodes. 
For the distance between marker nodes $x$ and $y$, we calculated the number of clusterings in which $x$ and $y$ were not in the same cluster.
 Using the matrix of pairwise distances, we produced a 2-dimensional visualization using metric MDS.  See supplementary materials for additional details. 

\section{Results \& Discussion} 

 \subsection{Properties of the citation network} Exosomes are an area of investigation within the larger field of extracellular vesicles in biology \citep{raposo2021,kalluri2020}. This field has been exponentially expanding, as evidenced by a keyword search for ``exosome" in the Dimensions bibliography yielding a count of less than 100 publications in 1990 and earlier, 11,100 publications from 1991 through 2010, and 115,300 publications from 2011-2021 (rounded). 
 To find communities in this  research area, we constructed and analyzed a large citation network that captured articles concerning exosomes and extracellular vesicles. 

 The citation network we built consisted of 14,695,475 publications in 13 components, of which the largest component accounts for 99.998\% of the network. The network was generated  through amplifying a seed set of 11,156 articles (Materials and Methods). 
 The degree distribution of the nodes in this network is typical of citation networks with a few nodes of high degree and many nodes of low degree (Figure \ref{fig:degree_distribution}).
 Roughly 68\% of the nodes have degree at most five and the 90th, 95th, and 99th percentiles of degree counts are 6113, 9186, and 24510.
 The highest degree is 256,836 and corresponds  to an article describing an assay for protein measurement.

The nodes in this dataset were roughly distributed by year of publication as follows-- 1990 or earlier (1.17 million), 1991-2010 (6.05 million), and 2011-2021 (6.99 million), suggesting not only a rapid growth of exosome publications in the post-1990 period but also a substantial increase in publications linked by citation to the seed articles. 

\begin{figure}[H]
\centering
\includegraphics[width=0.6\linewidth]{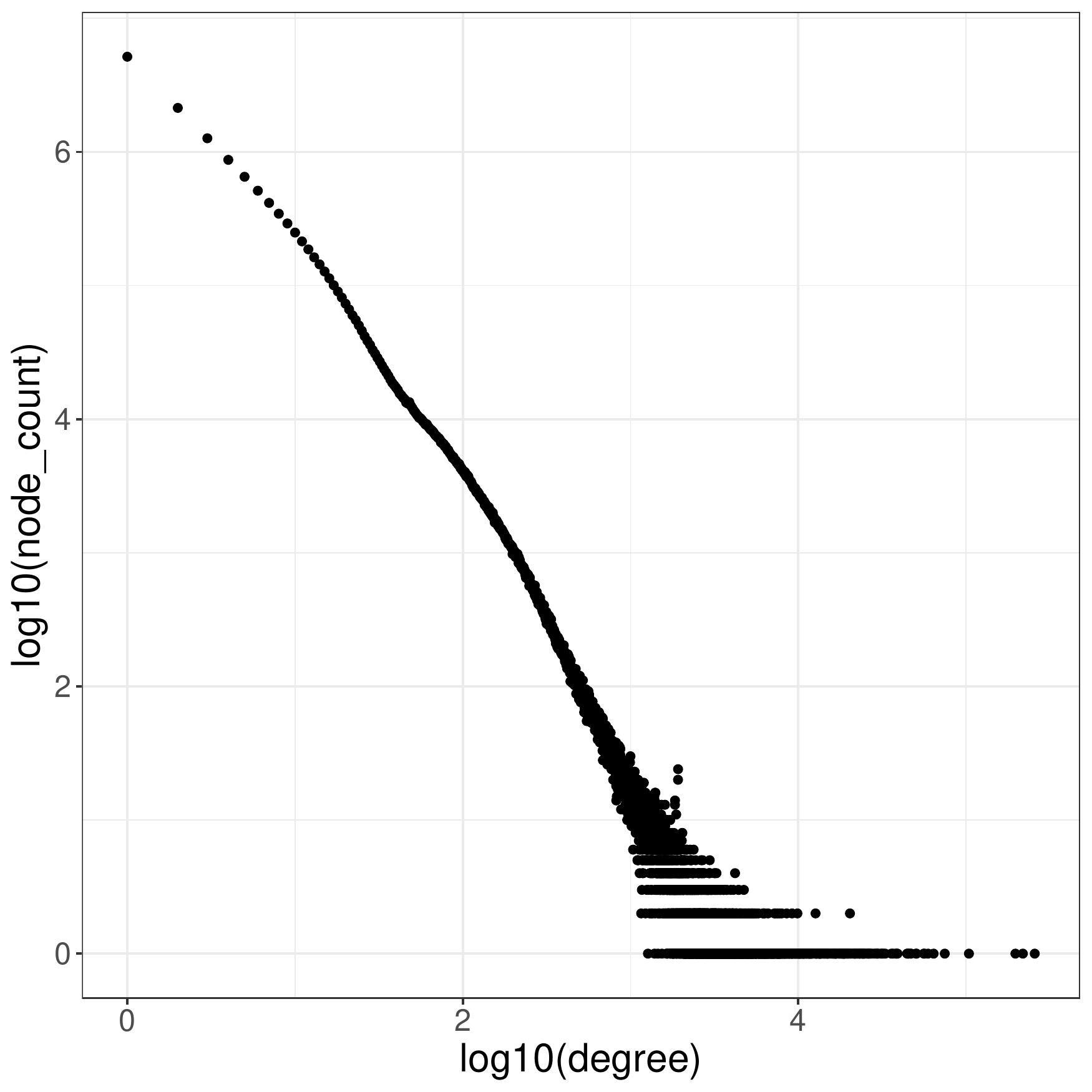}
\caption{Intra-network degree distribution of the 14,695,475 nodes in the exosome citation network. \emph{x-axis:} log degree, \emph{y-axis:} log node count.}
\label{fig:degree_distribution}
\end{figure}

\begin{figure}[H]
\centering
\includegraphics[width=0.7\linewidth]{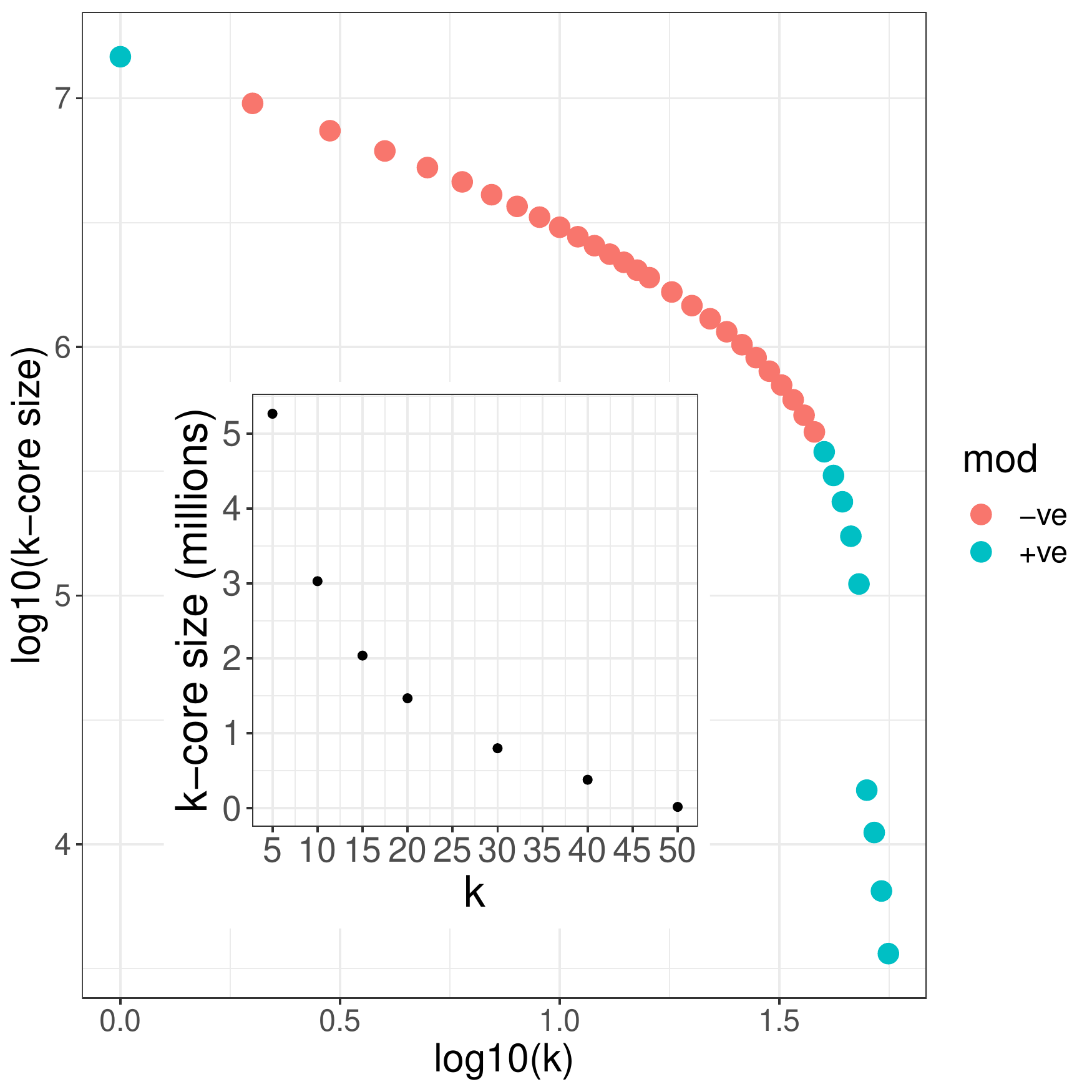}
\caption{Cluster sizes using the simple $k$-core clustering algorithm. The input exosome citation network (Materials and Methods) consists of 13 components summing to 14,695,475 nodes. Of these, a single component accounts for 14,695,226 articles (99.998\% of the network).  The $k$-core clustering algorithm was applied to the exosome citation network for multiple values of $k$ \emph{(x-axis)}. The size of the $k$-core is shown on the ordinate \emph{(y-axis)}.  
A single component is returned in each case, with the exception of $k$=1 and $k$=2. Node coverage decreases as $k$ increases, and is 
approximately 36\% and 21\% when $k=5$ and $k=10$ respectively.  The single cluster at $k$=56 has 3,630 nodes, amounting to 
node coverage of 0.02\%.  At $k$=57 or greater, the entire network dissolves, thus the degeneracy of the network is 56.  Clusters shown 
have positive modularity (\emph{mod +ve}) when $k=1$ or $k >=40$ and are colored teal.  By definition, the connected components of the network are the $1$-cores so 13 clusters are returned for $k=1$.  For $k=2$, only two clusters are returned. Clusters of size 100 or less (0.001\% of the network) are not included in this plot and pertain only to $k$=1 or $k$=2.
 \emph{Inset:} $k$-core sizes sampled at intervals of 5 are displayed using a linear scale.}
\label{fig:kcore}
\end{figure}

\subsection{Results of clustering methods}
We now present results using different clustering methods, including the $k$-core clustering method, Iterative $k$-core clustering, our four-stage pipeline, and the Leiden algorithm. As noted earlier, we were explicitly interested in discovering citation-dense regions with center-periphery structure that reflected cohesiveness and collaboration \citep{Giatsidis2011,Breiger2014}. In this case, collaboration refers to the recognition of prior work by others in the community through citations. We were not interested in communities that consisted of a single heavily cited article or a single article that cited many references \citep[p.~193]{Chandrasekharan2021}. 

\subsubsection{Results for the Leiden algorithm} We first used the Leiden algorithm, which guarantees connected communities \citep{Traag2019}, as a benchmark for community finding. We consider Leiden as a reference method in the scientometrics community, and  have previously used it \citep{Chandrasekharan2021} in combination with Markov Clustering \citep{VanDongen2008} to detect communities in the immunology and ecology literature. 
In the present study, we used the Leiden software~\citep{leiden-github} with the Constant Potts Model as  quality function, and with varying resolution factors to cluster the exosome citation network. The resolution factor is designed to modify the clustering, as it determines  the required minimum density within communities. 

For this citation network,  at most of the resolution values that were tested,  the Leiden algorithm generated a relatively large number of small clusters (Table \ref{tab:leiden}).
At the smallest resolution factor value we examined (0.05), Leiden produced 488,285 non-singleton clusters and 6,323,695 singletons and the largest cluster comprised 960 nodes. Node coverage, defined as the ratio of nodes in non-singleton clusters to the total number of nodes in the network was 57\% with resolution factor set to 0.05.  As we increased the resolution value,  the number of singleton clusters increased, the sizes of non-singleton clusters decreased, and we observed a progressive drop in node coverage (down to 16\% at resolution factor 0.25 and to 9\% for resolution factor 0.50).  

\begin{table}[ht]
\centering
\begin{tabular}{c|c c r c cc }
    \hline
    Resolution & Node Coverage & \# Clusters & \# Singletons & Min &  Median & Max  \\
    \hline
    0.05 & 57\% & 488,285 & 6,323,695 & 2 & 17 & 960 \\
    0.25 & 16\% & 481,780 & 12,408,446 & 2 & 4 & 192 \\
    0.50 & 9\% & 434,973 & 13,412,183 & 2 & 2 & 97 \\
    0.75 & 8\% & 489,937 & 13,496,111 & 2 & 2 & 39 \\
    0.95 & 8\% & 497,757 & 13,499,132 & 2 & 2 & 16 \\
    \hline
\end{tabular}
\caption{\textbf{Cluster statistics for Leiden under different resolution values.}  Node coverage, expressed as percentage, refers to the the ratio of nodes in non-singleton clusters to the total number of nodes in the network. \# of clusters refers to the number of non-singleton clusters. The last three columns refer to the sizes of non-singleton clusters. Results shown here are before $kmp$-processing.}
\label{tab:leiden}
\end{table}

We also screened the clusters generated by Leiden for $kmp$-validity (Materials and Methods, Definition \ref{def:valid}) at $k$=5 or 10 and $p$ = 2, the values that we use in when applying our pipeline. 
For the resolution  value of 0.05, the node coverage 
is 3.16\% and 2.54\% when $k$ is 5 or 10 respectively (Supplementary Materials, Tables 3 and 4), a large drop from 57\% for the
Leiden clustering before restriction to $kmp$-valid clusters.
For larger resolution values, restriction to $kmp$-valid clusters results in even smaller node coverage.
In addition, again using resolution value $0.05$,  after restriction to $kmp$-valid clusters, the number of clusters is much smaller: only 4,076 for $k=5$ and 2,320 for $k=10$, a large drop from 488,285. 
Thus, restriction to $kmp$-valid clusters greatly reduces the node coverage and number of clusters compared to Leiden before
restriction, and this holds for all resolution values.
Moreover, this analysis shows that  only a small fraction of the clusters produced by Leiden at any resolution value  are $kmp$-valid (e.g.,
less than 1\% for resolution value 0.05).

In summary, while Leiden was able to efficiently cluster our network into a large number of small communities that represented between 8\% and 57\% of the network depending on the resolution factor employed,  only a small fraction of the communities exhibited $kmp$-validity. This is not surprising since the Leiden algorithm and optimization criterion was not designed to produce $kmp$-valid communities, and these trends indicated that Leiden had limited utility for our purposes.

\subsubsection{Results for $k$-core clustering} 
To identify citation-dense regions of the network, we ran the $k$-core clustering method on our network for different values of $k \geq 1$ (Figure \ref{fig:kcore}).  
The largest value for $k$ for which a cluster is returned is $56$,  indicating that the degeneracy of the network is $56$.
The network has no isolated vertices and so $k=0$ and $k=1$ produce the same output, which is 13 clusters, 
each corresponding to a connected component in the network.
When $k=2$, two clusters were returned. 
For each $3< k \leq 56$, only a single cluster was returned, hence only $k=0,1,2$ produced more than one cluster, and in each of these cases, a single cluster dominated in size.

An interesting trend in this analysis is how $k$ impacts   the modularity scores of the clusters.  
By definition, the components in the graph all have positive modularity, and so
when $k=1$ the clusters all have positive modularity. For larger values of $k$, the modularity scores do not become positive until $k=40$, and then all subsequent values of $k$ produce clusters with  positive modularity scores.

Cluster size decreased monotonically  as $k$ increased to $56$, however we observe a pattern of relatively stable core sizes at lower values of $k$  followed by a more rapid decrease as $k$ increases above 40.  
 \cite{Leskovec_2008} reported similar findings about changes in core sizes on a much larger network of instant messaging data that consisted of 180,000,000 nodes. These authors suggest that the rapid decrease in core size occurs once nodes on the fringe of the network are removed.
On our network, this more rapid decrease of cores sizes also coincides with the appearance of positive modularity of clusters; modularity was not 
reported in \cite{Leskovec_2008}. Thus, we are mainly interested in the $k$-cores for large values of $k$:   they have positive modularity, small changes in $k$ result in large changes to their sizes, and they have been identified in the prior literature as what is left after the fringe of the network is removed. 

While the simple $k$-core clustering method identifies citation-dense areas in the network suggesting cohesiveness and collaboration, it has two limitations: it does not ensure positive modularity for every cluster it produces, and, on our data, for every $k \geq 3$ it produced only a single large cluster. We note that for $k=0,1,2$, it produced  two or more clusters, one of which was very large.  These limitations impact the ability to find multiple communities of interest in the exosome literature,  especially considering the possibility that some of the communities of interest could be contained within larger clusters with non-positive modularity. 

\subsubsection{Results for Iterative k-core (IKC)} 
We designed 
Iterative $k$-core (IKC) to improve on the $k$-core clustering algorithm.
The input to IKC is
 the network $N$ and the parameter $k$.
The first cluster that is found in the network is the $L$-core where $L$ is the largest label
assigned to any node. If $L \geq k$, then the $L$-core is produced as a cluster and  removed from the network, and otherwise
the algorithm stops.
If the algorithm has not stopped, it is run recursively on the reduced network.
Therefore, IKC($k$) will contain all the clusters in IKC($k'$) for $k \leq k'$.

We explored IKC varying $k$ between $5$ and $50$, and examined the distribution of cluster sizes generated. We also recorded the minimum degree in each cluster. Since IKC clusters  satisfy $k$-validity, the nodes in the clusters are all ``core'' members, and so this minimum degree is also the Minimum Core Degree (MCD) of the cluster.  
As expected, increasing $k$ results in decreases in cluster size and increases in the MCD value (Figure \ref{fig:ikc}). 
In order to include as much of the network as is reasonable and still have sufficient density to define community structure, we selected  $k=5$ and $k=10$ for IKC.

While IKC discovered $km$-valid communities that were trivially $p$-valid, it too has limitations. It identifies only core nodes with modest node coverage of 7.38\% and 4.22\% at $k$=5 and $k=10$,  respectively. 
It also generates some large clusters with lower MCD values that leave open the question of whether denser $kmp$-valid communities exist within them.

\begin{figure}[H]
\centering
\includegraphics[width=0.7\linewidth]{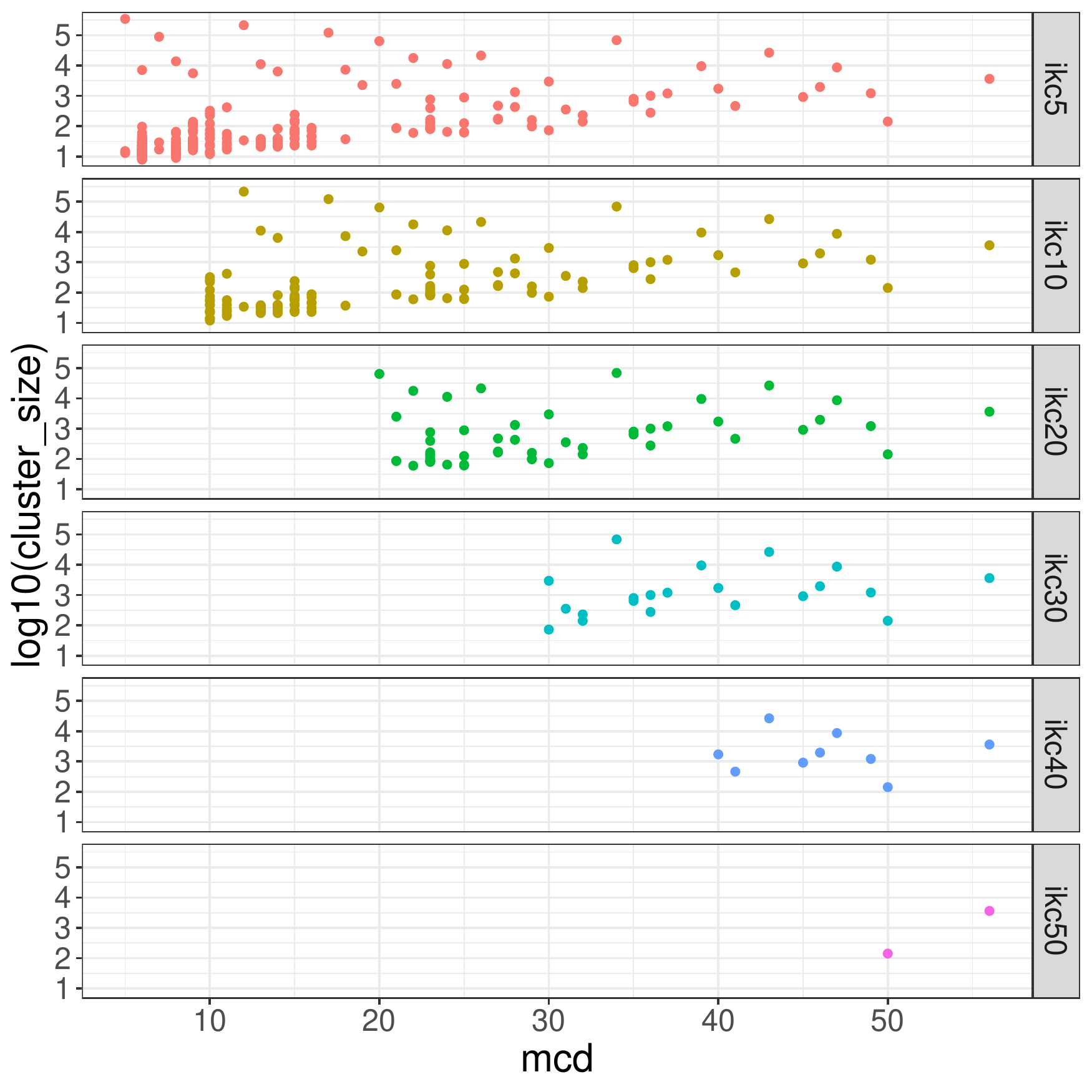}
\caption{Empirical statistics of the Iterative k-core (IKC) clustering methods (varying $k$) on the Exosome network. The y-axis shows the cluster sizes (logarithmic scale) and the x-axis shows the Minimum Core Degree (MCD) values, where the MCD of a cluster is the minimum degree of any node in the cluster. By design,  IKC($k$) contains all the clusters of IKC($k'$) if $k \leq k'$; thus, the panels showing IKC at lower values of $k$ contain greater numbers of clusters.}
\label{fig:ikc}
\end{figure}

\subsubsection{Results for the Four-Stage Clustering Pipeline}

\begin{figure}[H]
\centering
\includegraphics[width=0.9\linewidth]{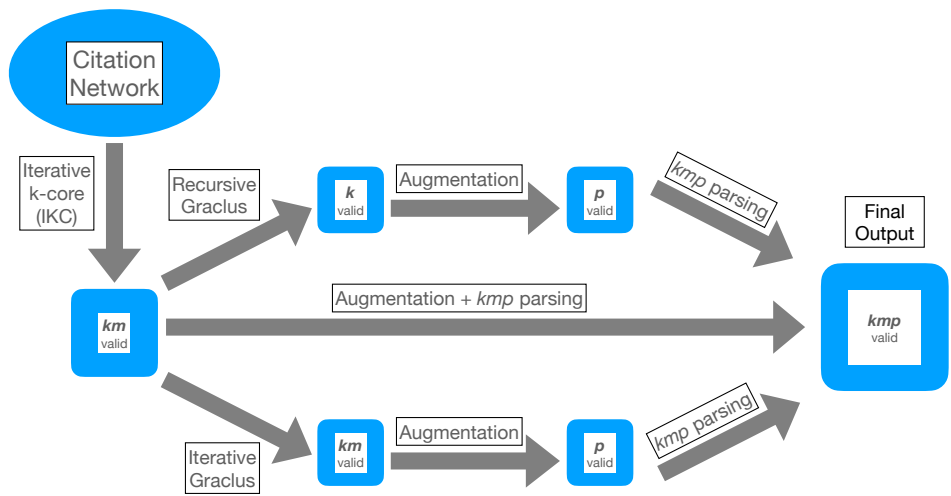}
\caption{Our four-stage clustering pipeline takes as input a citation network and produces clusters based on  values selected for parameters $k$ and $p$. 
Boxes adjacent to edges indicate stages in the pipeline; boxes with blue borders indicate tests that are performed to determine which clusters are passed to the next stage.
In Stage 1, it runs the Iterative k-core (IKC) algorithm for the selected value of $k$; clusters that are $km$-valid are then passed to the next stage.  Stage 2 (optional) breaks the clusters from the first stage into smaller clusters, using either recursive Graclus or Iterative Graclus (and in each case, using either the default version or a heuristic search version). Clusters that pass the required validity check
 ($k$-validity for Recursive Graclus and $km$-validity for Iterative Graclus) are then passed to the next stage.  These clusters are then enlarged with additional nodes  in Stage 3, the ``Augmentation" step. All clusters produced are $p$-valid at this point, and are passed to Stage 4, $kmp$-parsing, which produces a set of $kmp$-valid clusters, each of which is parsed into their core subcluster and non-core subcluster.   }
\label{fig:flow_chart2}
\end{figure}

To address the limitations of IKC, we designed a four-stage pipeline (Figure \ref{fig:flow_chart2}). This four-stage pipeline is guaranteed to produce a $kmp$-valid clustering of the network, for user-provided values of $k$ and $p$. 

In Stage 1 we use IKC with $k=5$ and $k=10$. 
Stage 2 breaks the clusters found in Stage 1 into smaller clusters; for this stage we apply either Recursive Graclus or Iterative Graclus, and for each of these we use Graclus either in default mode or in local search mode. 
Thus, for each setting of $k$ we have four versions of the four-stage pipeline: Iterative and Recursive Graclus run in either default or local search mode. In Stage 3, the clusters produced by Stage 2 are  augmented by the addition of nodes that satisfy  $p$-validity for  $p=2$.
In Stage 4, we parse the output of Stage 3 and retain only those clusters that are $kmp$-valid.

\begin{table}[h]
\centering
\begin{tabular}{lc|c r r r r r }
    \hline
    Stage 2  &  Stage 3 & Node Coverage &  Clusters &  Singletons & Min &  Median & Max  \\
     &  & &(number)  & (number)  & & & \\ 
    \hline \\
    {\bf $k=5$ }\\
    No & No & 7.38\% &  276  & 13,611,485 & 8 & 28.0 & 345,139 \\
    No & Yes & 36.63\% & 276 & 9,312,583 & 15 & 265.5 & 856,623 \\
    IG(0) & Yes &  20.01\% & 13,709 & 11,752,582 & 6 & 118.0 & 19,934 \\
    IG(2000) & Yes & 20.84\% & 9,698 & 11,632,959 & 6 & 106.0 & 32,487 \\
    RG(0) & Yes & 26.27\% & 2,261 & 10,835,596 & 10 & 145.0 & 579,720 \\
    RG(2000) & Yes & 26.09\% & 3,417 & 10,861,670 & 11 & 105.0 & 578,265 \\
    \hline \\
{\bf  $k=10$}\\
No & No & 4.22\% & 119  & 14,075,787 & 12 & 85.0 & 213,670 \\
	No & Yes &  33.69\% & 119  & 9,744,368 & 62 & 1638.0 & 964,503 \\
    IG(0) & Yes & 18.51\% & 4,185 & 11,975,796 & 33 & 427.0 & 27,007 \\
    IG(2000) & Yes & 20.00\% & 3,044 & 11,757,196 & 31 & 488.5 & 60,189 \\
    RG(0) & Yes &  27.26\% & 359 & 10,689,730 & 67 & 1014.0 & 679,922 \\
    RG(2000) & Yes &  27.61\% & 473 & 10,637,388 & 55 & 761.0 & 620,491 \\
    \hline
		 
\end{tabular}
\caption{\textbf{Cluster statistics for 12 variants of Four-Stage clustering.} All results include Stages 1 and 4, 
 but some pipelines do not use Stages 2 or 3; all clusterings are $kmp$-valid. 
Top: results for $k=5$, Bottom: results for $k=10$.  
Stage 2 is performed either using Iterative Graclus (IG) or Recursive Graclus (RG), which are each run with either $0$ or $2000$  local search iterations.  Node coverage refers to the percentage of network nodes contained in non-singleton clusters and singletons refers to the number of nodes in singleton clusters. All other statistics refer to non-singleton clusters, with the last   three columns refering to the sizes of non-singleton clusters. 
Specific noteworthy trends include:  (a) all clusterings that use Stage 3 have node coverage above 18\%, (b) all clusterings have at least one very large cluster, (c) Stage 2 choice impacts maximum cluster size, and (d) setting $k=5$ produces more clusters with a smaller median cluster size than setting $k=10$.}
\label{tab:ig-rg-cluster-stats}
\end{table}

We show results for the different versions of this 
four-stage pipeline in Table \ref{tab:ig-rg-cluster-stats}, where the 
rows  correspond to different ways of setting $k$ and running Stages 2 and 3.
Using IKC alone without Stages 2 and 3 
resulted in   276 and 119 non-singleton clusters and had total node coverage of 7.38\% and 4.22\%, with maximum cluster sizes of 345,139 and 213,670, for $k=5$ and $k=10$, respectively. 
Skipping Stage 2 (breaking down clusters) but adding Stage 3 (augmentation) greatly increases the node coverage to at least 33\% for both settings of $k$ but also increases the maximum cluster size  to 856,623 and 964,503 for $k=5$ and $k=10$, respectively.
This approach also has large median cluster sizes, especially for $k=10$ where the median is 1638. 
Thus, using Stage 3 (augmentation)  but not also Stage 2 produces high node coverage and large cluster sizes.  

Using all four stages,  and hence using Stage 2, reduces the node coverage to values that range from 18.51\% to 27.61\% and also reduces the median and maximum 
cluster sizes, but the choice of how Stage 2 is run has a significant impact.  
Node coverage is higher when using Recursive Graclus rather than Iterative Graclus.
Using Recursive Graclus in default mode rather than local search mode
tends to produce a smaller number of non-singleton clusters that are also somewhat larger; for example, when $k=10$, default usage of Graclus that doesn't employ the local search strategy produces a median cluster size   of 1,014 as
opposed to  761  when using 2000 local search iterations.

The choice between $k=5$ and $k=10$ also impacts results, with
$k=10$ producing a much smaller number of non-singleton clusters that are substantially larger than the results for $k=5$. 
For example, using default Recursive Graclus at $k=10$ produces 359 non-singleton clusters
with median size 1014 while the same setting for $k=5$ produces 2,261 non-singleton clusters with median size 145.

\subsubsection{Comparison between different clustering methods}
We now compare the  different clustering outputs with respect to node coverage, 
number of non-singleton clusters, and the distribution of non-singleton cluster sizes.
As our discussion of the different variants of the four-stage pipeline revealed,  
how Stages 2 and 3 are run and the value for $k$  impact these statistics.
If node coverage is the most important criterion, then skipping Stage 2 but using Stage 3 is 
recommended. However, since these approaches produce very large clusters, 
including Stage 2 is more likely to be desirable.
Among the techniques that use Stage 2, Iterative Graclus tends to produce smaller clusters, but Recursive
Graclus produces higher node coverage; the choice between these should be made based on the 
specific question that is being addressed.    
Similarly, how $k$ is set should depend on the features of the citation network, and picking larger
values of $k$ may be suitable under some conditions.

A comparison to Leiden is also helpful: before restricting to the $kmp$-valid clusters,  
Leiden has very high node coverage (57\%)  for resolution value $0.05$, but after restricting to $kmp$-valid clusters, node coverage drops to 3.16\% and 2.54\% when k is 5 or 10, respectively. In contrast, the  four-stage pipeline has node coverage that varies from 18.51\% to 27.61\%, depending on $k$ and how Stage 2 is performed (Table  \ref{tab:ig-rg-cluster-stats}).
 
Thus, a potential user of this approach is presented with options that can be used to address contextual needs. For example, after considering features of the source data and the purpose of community finding, a user may choose Leiden, IKC, IKC with augmentation, or the complete pipeline with its $kmp$-parsing requirements.

\subsection{Marker Node Analysis} The network we constructed for this study has more than 14 million nodes. 
We assumed that some of the communities discovered would represent  areas of investigation peripheral to exosomes and extracellular vesicles. 
Consequently, we used a set of 1,218 independently selected articles from the extracellular vesicle literature, all of which are present in the network, as marker nodes (Materials and Methods) and used them to identify clusters of interest.  
 Any community containing at least one marker node was considered relevant; however, we were particularly interested in communities with high numbers of marker nodes.

For IKC clustering at either $k$=5 or $k$=10 (Figure \ref{fig:kcore}), two clusters contained 256 and 227 marker nodes respectively, and together accounted for approximately 40\% of the 1,218 markers. The first of these clusters, with 256 marker nodes, was the 56-core of the network, and comprised 3,630 nodes. 
The second, with 227 marker nodes, exhibited an MCD (minimum core degree) value of 12 and consisted of 213,670 nodes.
 To visualize the distribution of the entire set of marker nodes, we used a multidimensional scaling approach using the frequency of co-occurrence in clusters across 12 different clustering methods as the measure of similarity between publications. 
 Each of these two sets of marker nodes are found in dense and clearly defined clusters after multidimensional scaling (Figure \ref{fig:mds}).

\begin{figure}[H]
\centering
\includegraphics[width=0.9\linewidth]{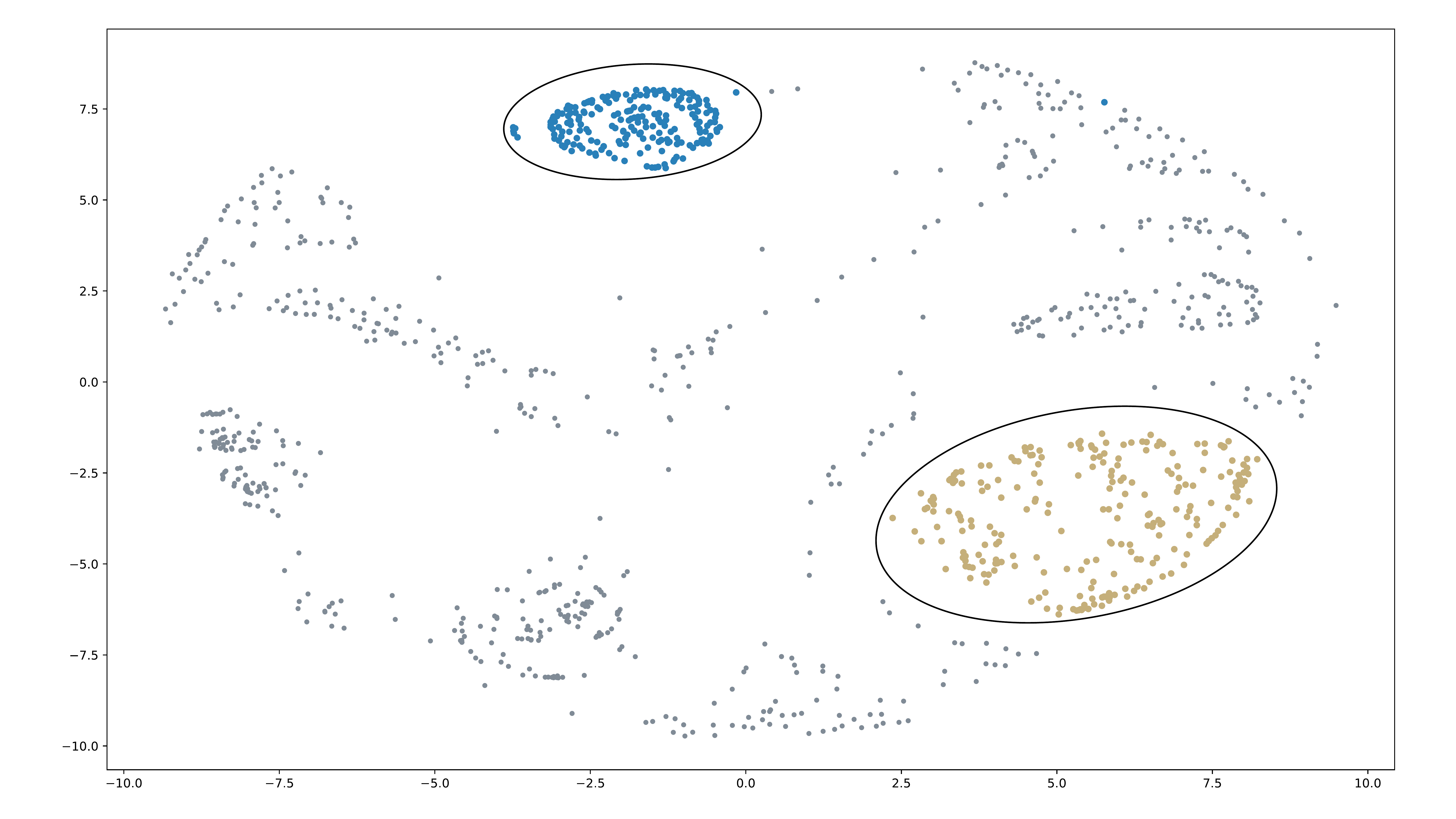}
\caption{Multidimensional scaling (MDS) of 1,218 marker nodes based on the frequency with which they are placed in the same cluster in 12 clustering outputs.  The circled beige cluster (bottom right) corresponds to the 256 markers found using either IKC(5) or IKC(10)  in a single cluster of minimum cluster degree (MCD) of  56, while the circled blue cluster (top) corresponds to the 227 markers found in a different cluster with an MCD value of 12. }
\label{fig:mds}
\end{figure}

 These two clusters were similar in having a large number of marker nodes but were otherwise different from each other with respect to MCD and size; therefore, we used the two sets of 256 and 227 markers as examples for further study. 

A second criterion considered was robustness to clustering method, which we measure by the frequency with which marker nodes were found co-located in the same cluster across the 12 clusterings we studied. 
We began this evaluation by first determining which of these nodes are always placed in non-singleton clusters for all 12 clustering methods.
We found that 27 of 256 and 35 of 227 marker nodes respectively were  always placed in non-singleton clusters, and studied these two sets further.

The first set (A) consisting of 27 marker nodes contained 17 articles and 8 reviews that were published between 2006 and 2019 with 23 of these published in 2010 or later. 
Based upon inspection of titles and journal, the contents of articles in  set A spanned basic cell biology, the role of exosomes in cancer, and exosome isolation methods with some variation in terms of being descriptive or mechanistic. 
The second set  (B) of 35 markers contained 26 articles and 9 reviews published between 2013 and 2021, largely focused on a basic and translational studies of exosomes in nervous tissue but also including a few articles on exosomes in pregnancy and exosome isolation methods. 

We then examined the publications in sets A and B for co-occurrence in the same cluster across all 12 clusterings.
We  found that 8 articles from set A were always found in the same cluster and 4 articles from set B were always found in the same cluster.
While the numbers of markers in these sets are small, they serve to identify a larger community, which can then be characterized further by other techniques, such as detailed scholarly examination or textual content analysis.

We identified the smallest cluster across the 12 clusterings that contained the 8 articles from A; the selected cluster  (Cluster 1) consisted of 73 articles and was focused on extracellular vesicles in cancer. We performed a similar analysis for the 4 articles from B, and found a cluster (Cluster 2) of 145 articles that was focused on extracellular vesicles in the nervous system. 
Thus, both these clusters were relatively homogenous with respect to article content, one focused on extracellular vesicles in cancer and the other on extracellular vesicles in the nervous system. Our use of a single label for each cluster should be considered a subjective approximation; an alternate view  is that Cluster 2 also includes articles on cancer and has some focus on astrocytes. Both clusters were  derived from the IKC(5)-Iterative Graclus branch of the pipeline.

We then extracted the authors of articles in these two clusters (Table \ref{tab:author-stats}). Cluster 1 involved 356 authors of which 9 were  
authors of at least 5 articles in the cluster and one person was an  author of 17 articles in the cluster. 
However, 301 authors (84.6\%) had contributed to only one article in the cluster. 
Cluster 2  involved 742 authors of which 4 were authors of at least 5 articles  in the cluster.
One author had contributed 7 articles and 650 authors (87.6\%) had contributed only one article each. Interestingly, the two publication clusters share 14 authors. Thus, the two author communities defined by two disjoint publication clusters  overlap.

 These trends are strikingly similar to the observations of \cite{Price1966} in that the authors segregate into a small number with large numbers of papers  in the cluster and many with only one paper in the cluster. Although we only examined two clusters, both exhibited the center-periphery structure described in \cite{Price1966}; our sample is too small to draw conclusions beyond suggesting that the trends observed may be true for other clusters, and this should be evaluated.
 
We also found discrete co-author groups within these clusters (Figure \ref{fig:coauthors}). Cluster 1 featured 4 non-overlapping co-author groups where authors were linked to each other if they had co-authored at least two articles in the cluster. Cluster 2 featured 17 such discrete co-author groups suggesting, despite its larger size, that influence within the group was more distributed considering the larger number of co-author groups and the smaller number of articles written by individual authors. These examples are provided to illustrate the potential utility of the pipeline and the use of marker nodes. 

\begin{table}[ht]
\centering
\begin{tabular}{rcccccc}
  \hline
 & Articles & Core Nodes & Markers & Authors & Auth\_$5$  & Max\_Pubs  \\ 
  \hline
Cluster 1 & 73 & 47 & 16 & 356 & 9 & 17 \\ 
Cluster 2 & 145 & 129 & 31 & 742 & 4 & 7 \\ 
\hline
\end{tabular}
\caption{\textbf{Statistics on the two selected small clusters}. Cluster 1 and Cluster 2 were selected based on marker nodes: the  greatest number of marker nodes that were found co-located in the same cluster across all 12 clustering methods were used to define their common clusters, and the smallest of the common clusters for each set was returned.   Both clusters are $kmp$-valid for $k=5$ and $p=2.$ 
Articles refers to the total number of nodes in the cluster, core nodes refers to the number of these nodes that are core, and markers refers to the
number of nodes that are markers. 
Authors denotes the total number of authors of articles in the cluster, 
Auth\_$5$ refers to the number of authors that have at least 5 publications in the cluster, and Max\_Pubs refers to the largest number of publications in the cluster by any single author.}
\label{tab:author-stats}
\end{table}

\begin{figure}
    \centering
    \begin{subfigure}[t]{0.48\textwidth}
        \centering
        \includegraphics[width=0.8\linewidth]{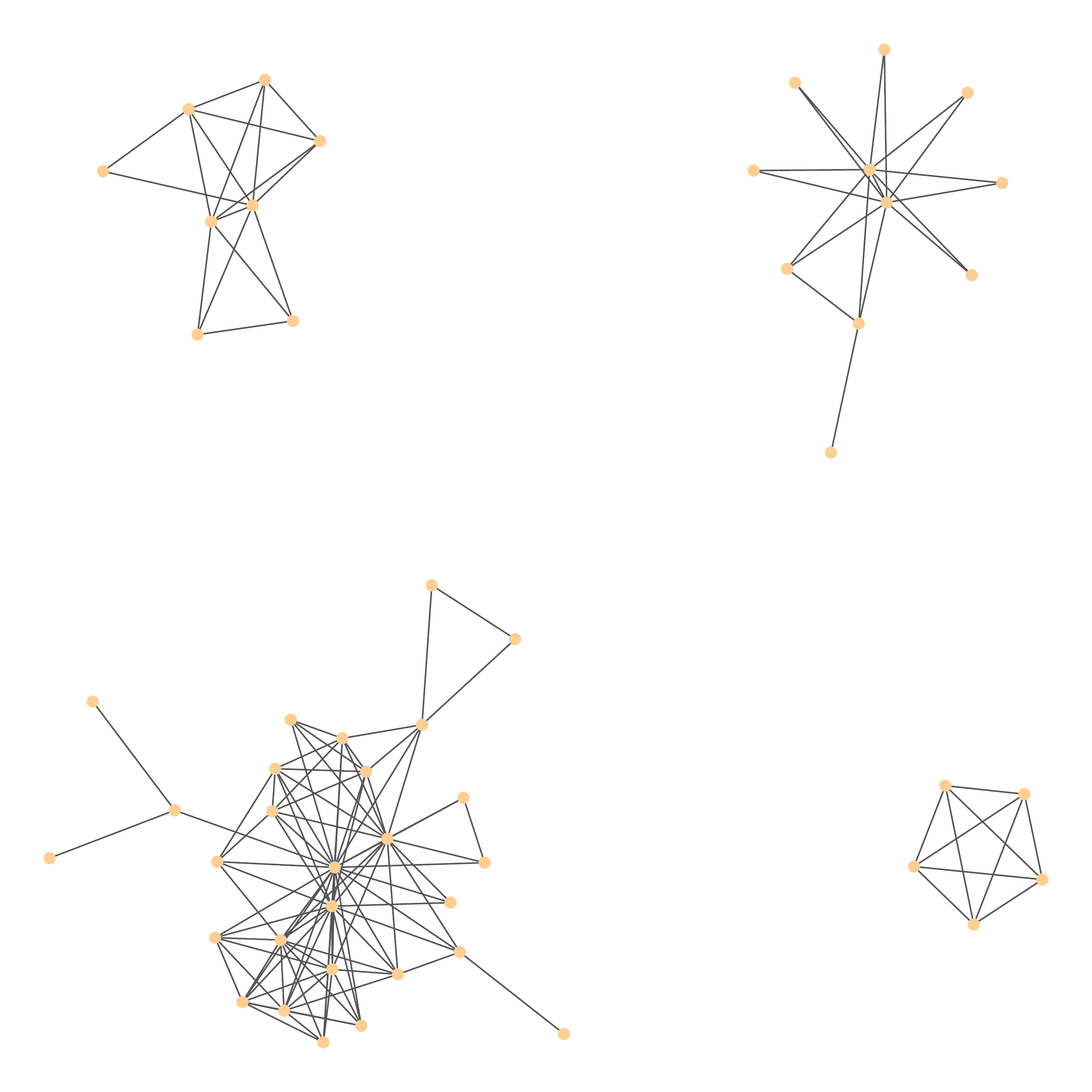} 
        \caption{Cluster 1} \label{fig:timing1}
      \end{subfigure}
    \hfill
    \begin{subfigure}[t]{0.48\textwidth}
        \centering
        \includegraphics[width=0.8\linewidth]{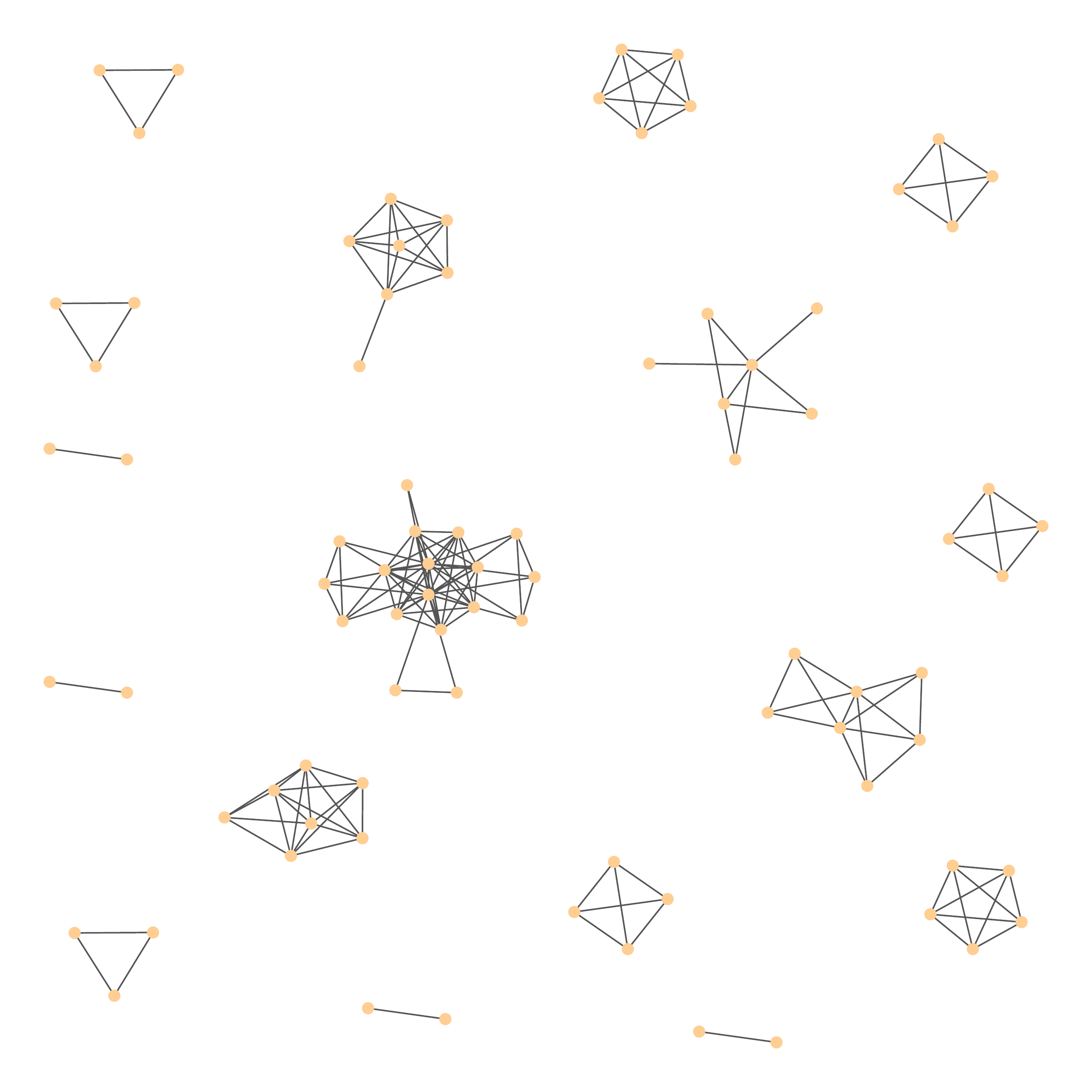} 
        \caption{Cluster 2} \label{fig:timing2}
    \end{subfigure}
 \caption{\textbf{Co-author clusters in two communities identified using marker nodes.} Two clusters were selected for analysis because they contained the greatest number of marker nodes that were co-located  across all 12 clustering methods, and were the smallest of such clusters (see text). Discrete co-author groups are found in both clusters when inclusion in a co-author groups requires at least two instances of intra-cluster co-authorship between two authors. \emph{(a)}  Cluster 1 consists of  73 articles contributed by 356 authors. This cluster contains 4 non-overlapping co-author groups, with 5, 8, 11, and 28 authors. 
 \emph{(b)} Cluster 2 consists of 145 articles authored by 356 authors. This cluster contains 17 non-overlapping groups, with four groups of 2 authors, 
 three groups of 3 authors, three groups of 4 authors, two groups of 5 authors, four groups of 7 authors, and one group of 18 authors. }
 \label {fig:coauthors}
 \end{figure}
\section{Conclusions} 

Based on historical studies of research communities, we posed corresponding properties for the graphical structure of communities in networks.  We developed an analytic  pipeline to ask whether communities of publications with center-periphery substructure exist in citation networks.
In designing a four-stage pipeline to find communities of this form, we were implicitly asking whether the information encoded in the graphical structure of communities can be used to make inferences on the social structure of these communities, for example, discrete co-author groups. We examined these questions using a citation network representative of exosomes and extracellular vesicles, a field that has rapidly expanded in recent years. 

In this citation network,  our pipeline found many publication communities that exhibit center-periphery structure.
This finding supports our hypothesis that communities of this type exist within the extracellular vesicles research community,
and shows that the pipeline we used can find such communities. 
Whether such communities exist in other citation networks and whether our pipeline is successful at finding such communities are important
questions that future work should address.

Our pipeline is designed to enable investigators to interrogate their data with different options for each stage  and different settings for $k$ and $p$. 
As we observed in our study, changes to the settings for the parameters $k$ and $p$ as well as how each stage was performed produced  clusterings that differ from each other in terms of the the node coverage as well as the number and sizes of non-singleton clusters, effectively providing different views of the network. 
Thus, the specific question of interest and the properties of the citation network are important in choosing how to set these parameters.
Alternatively, the pipeline can be used to generate many different clusterings, and the investigator may assess community structure through an integrative analysis that does not depend on a single clustering method.

We also saw significant differences between clusterings produced by our pipeline and those produced by the Leiden software. 
While there is some overlap in the range of cluster sizes generated, our pipeline tends to generate much larger clusters than the Leiden software, and all of our pipeline clusterings produced at least one very large cluster with more than 19,000 nodes and, in the case of Recursive Graclus, the largest cluster contained more than 500,000 nodes.  Given our focus on the small clusters produced in our various pipelines, we did not explore large  publication clusters. It is not clear to us what insights can be gained from clusters that have tens of thousands of nodes. We speculate that they may reflect the many connections in a rapidly expanding field, and we also consider them as tempting targets for future versions of Stage 2, which seeks to break up the large clusters. Further work is clearly needed.

This study suggests several directions for future work. 
For method development, 
our four-stage pipeline is designed to enable substitutions to how each stage is performed; as noted above,  breaking up large clusters might be more successfully executed using new approaches rather than either recursive or iterative Graclus. 
We also note that all the clustering methods we developed produce disjoint clusters, yet  publications may be expected to be members of more than one research community.  Some clustering methods have been developed that can produce overlapping clusters  \citep{Rossetti2020}, and exploring this approach in the context of large citation networks is likely to provide additional insights. Methods that can combine information from multiple clustering methods  could lead to better insight into community structure, and so principled development of ensemble methods (a standard approach in machine learning) is another direction for future research.

One of the most intriguing directions  for future research is  the life cycle of these research communities, both in terms of how ideas and questions being focused on by the research community as a whole change over time and how authors move between different communities over time. This is challenging to study because the network evolves through growth each year and community finding in dynamic graphs is a promising direction \citep{Rossetti2020}. 

 Additional clusters could be examined for in-depth examination based, for example, on specific authors, articles of particular importance, funding sources, or clusters identified by marker node co-located in many but not all 12 clusterings. Other insights could be obtained by examining the relationship between publication communities (e.g., citations between communities) in a single clustering or comparisons of communities obtained using different clustering methods; such investigations would help elucidate how the research ideas and communities relate to each other, and the extent to which these communities are  hierarchically organized. 
 
 For extracellular vesicles, insight into community life cycles can be obtained by studying, over time, communities containing key studies such as the ones on transferrin recycling  \citep{harding1983,pan1983}, the observation that B-lymphocytes secrete antigen-presenting vesicles \citep{Raposo1996}, a report of exosome mediated transfer of mRNA and microRNA\citep{rata2006,Lotvall2007}, and the biological effects of transferring exosomes between lean and obese mice \citep{Olefsky2017}. Our analysis was based on a single set of marker nodes;  extending the set of marker nodes and annotating each marker node with respect to content would allow finer-grained evaluation. 
 
Exploration of author communities associated to these publication communities would shed additional insights into the social structure of the research community, potentially identifying authors with high influence within a particular emerging research area, and others that are highly influential across several areas.

We close with comments about the high-level approach we took to understanding community structure. 
 Our approach relied entirely on the graphical structure of the citation network. This restriction was used in order to provide a scalable approach that did not rely on any other information or expert knowledge; however, by design this limits which communities can be detected  \citep{McCain1986}.  Textual analysis, relationships between authors based on institutions, other social interactions such as conference presentations and sources of funding could be used to used to supplement citation data and would likely lead to a different set of publication or author communities with potentially different properties. 
 
 Understanding author role is also important, and again our reliance on citations to identify influential researchers is biased towards well-cited and well-funded authors. Perhaps one of the benefits, therefore, of our approach to community detection is that we can use it to find small and thematically-focused publication communities and hence identify those authors who are influential within these small communities. Nevertheless, we propose that while scalable methods for community detection may generally tend to rely on purely graph-theoretic properties, mixed method approaches support more a more nuanced understanding of social structures and dynamics within the scientific enterprise. 

\section*{Competing Interests} \vspace{3mm} The authors have no competing interests. Dimensions data were made available by Digital Science through the \href{http://www.dimensions.ai/scientometric-research/.}{free data access for scientometrics research projects program}. Digital Science personnel did not participate in conceptualization, experimental design, review of results, or conclusions presented. DK is an employee of NTT DATA, which had no role in this study.

\section*{Funding Information} EW is a Siebel Scholar.  TW receives funding from the Grainger Foundation. Research reported in this manuscript was supported by the Google Cloud Research Credits program through award GCP19980904 to GC.

\section*{Data Availability} Access to the bibliographic data analyzed in this study requires access from Digital Science. Code generated for this study is freely available from our Github site \citep{Park2021}.

\section*{Acknowledgments} We thank Valerie King from the University of Victoria for directing us to the k-core literature. We thank Phil Stahl from Washington University in St.~Louis for helpful discussions and for drawing our attention to recent reviews of the extracellular vesicle literature. We thank Digital Science, Google, the Grainger Foundation, and the Thomas and Stacey Siebel Foundation.   

\section*{ORCID IDs}
\begin{itemize}
\item Eleanor Wedell 0000-0002-7911-9156
\item Minhyuk Park 0000-0002-8676-7565
\item Dmitriy Korobskiy  0000-0002-7909-0218
\item Tandy Warnow 0000-0001-7717-3514
\item George Chacko 0000-0002-2127-1892
\end{itemize}


\end{document}